\documentclass[aps,pra,twocolumn,amsmath,amssymb,nofootinbib,superscriptaddress]{revtex4-1}

\usepackage[colorlinks]{hyperref} 
\usepackage[pdftex]{graphicx}
\usepackage{mathrsfs}
\usepackage{amsthm}
\usepackage[usenames,dvipsnames]{xcolor}
\usepackage{framed}
\usepackage{bbm}
\usepackage{mathrsfs}
\theoremstyle{definition}

\newcommand{\bra}[1]{\left\langle#1\right\rvert}
\newcommand{\ket}[1]{\left\lvert#1\right\rangle}
\newcommand{\ketbra}[2]{\left\lvert{#1}\middle\rangle\!\middle\langle{#2}\right\rvert}
\newcommand{\braket}[2]{\left\langle{#1}\middle\vert{#2}\right\rangle}
\newcommand{\tr}{\mathrm{Tr}}
\newcommand{\ptr}[2]{\mathrm{Tr}_{#1}\left( #2\right)}

\newcommand{\set}[1]{\left\lbrace #1\right\rbrace}

\newcommand{\expect}[1]{\langle #1\rangle}

\DeclareMathOperator*{\argmax}{argmax}

\definecolor{cream}{rgb}{1.0, 0.99, 0.82}
\definecolor{celadon}{rgb}{0.67, 0.88, 0.69}
\definecolor{beaublue}{rgb}{0.74, 0.83, 0.9}

\definecolor{shadecolor}{rgb}{1.0, 0.99, 0.82}

\usepackage{soul}

\hypersetup{
    colorlinks=true, 
    linktoc=all,     
    linkcolor=blue,  
    citecolor=blue,
    filecolor=blue,
    urlcolor=blue
}

\newcommand{\PhysMQ}{\affiliation{%
    Department of Physics and Astronomy,
    Macquarie University,
    Sydney NSW, Australia.}}
    
\newtheorem{definition}{Definition}
\newtheorem{proposition}{Proposition}
\newtheorem{theorem}{Theorem}

\newtheorem{corollary}{Corollary}

\newcommand{\EQuSMQ}{\affiliation{%
    Centre for Engineered Quantum Systems,
    Macquarie University,
    Sydney NSW, Australia.}}
    
\newcommand{\PhysSU}{\affiliation{%
    Department of Physics,
    Stockholm University,
    Stockholm, Sweden.}}
    
\newcommand{\PhysUQ}{\affiliation{%
    School of Mathematics and Physics,
    University of Queensland,
    St. Lucia QLD, Australia.}}

\newcommand{\EQuSUQ}{\affiliation{%
    Centre for Engineered Quantum Systems,
    University of Queensland,
    St. Lucia QLD, Australia.}}

\newcommand{\FAU}{\affiliation{%
Friedrich-Alexander University Erlangen-N\"urnberg (FAU), Department of Physics, Erlangen,Germany.}}

\begin{document}

\title{A Resource Theory of Quantum Measurements}

\author{Thomas Guff} 
\email[]{thomas.guff@students.mq.edu.au}
\PhysMQ \EQuSMQ \PhysSU
\author{Nathan A. McMahon} \FAU \PhysUQ \EQuSUQ \PhysMQ \EQuSMQ 
\author{Yuval R.\ Sanders} \PhysMQ 
\author{Alexei Gilchrist} \PhysMQ \EQuSMQ

\date{\today}

\begin{abstract}
Resource theories are broad frameworks that capture how useful objects are in performing specific tasks.
In this paper we devise a formal resource theory quantum measurements, focusing on the ability of a measurement to acquire information. 
The objects of the theory are equivalence classes of positive operator-valued measures (POVMs), and the free transformations are changes to a measurement device that can only deteriorate its ability to report information about a physical system.
We show that catalysis and purification, protocols that are possible in other
resource theories, are impossible in our resource theory for quantum
measurements. Standard measures of information gain are shown to be resource monotones, and the resource theory is applied to the task of quantum state discrimination.
\end{abstract}

\maketitle

\section{Introduction}

Intuitively it is clear that some measurements are more useful than others. 
If the primary task of a quantum measurement is to gain information about a system, a projective measurement should be considered more useful than rolling a die to obtain the `outcome'. 
So perhaps it is reasonable to expect that the usefulness, or resourcefulness, of a given measurement should be able to be quantified for a broad range of tasks. 
The natural setting for such a quantification is a resource theory \cite{Chitambar2019,2016coecke59,Fritz2015,Brandao2015}.

A resource theory is an agent-centric theoretical framework that characterises the possible transformations that can be performed `for free' on a system. 
Given a particular state of the system there are typically a limited number of other states which can be freely accessed. 
The limited transformations can introduce irreversibility, since if we transform from state $A$ into state $B$, it may not be possible to transform $B$ back into $A$. 
In this case $A$ is more resourceful than $B$, since $A$ can accomplish all tasks that $B$ can accomplish. There may be some states which cannot be irreversibly changed using free transformations. These are the least resourceful states, also known as the \emph{free states}. 

Given any two states of a system, it will not always be possible to freely convert one into the other. 
This can lead to interesting processes such as catalysis and purification, where we may exploit resourceful states to enable previously inaccessible transformations (either retaining or consuming the resourceful state in the process).
Two primary questions a resource theory aims to answer are about which resources can be interconverted, and what measures of resourcefulness can reflect and describe this. 
In general the resource theory reflects a pre-order or partial order between states, so a single measure of resourcefulness will not suffice. 
In these cases we will need a set of measures that describe the resourcefulness of the state to accomplish this second task.

In this paper, we develop a resource theory of quantum measurements. Our theory
is derived from two key operations a user could perform that should never improve the measurement:
an operation whereby two outcomes are \emph{confused}, and an operation whereby a redundant outcome can be \emph{made up}.

We show that this formulation leads to a resource theory that lacks
catalysis and purification, in contrast with more familiar resource theories such as entanglement.
To demonstrate the validity of our approach, we further show that extant
measures of the information gain of a measurement are monotones in our theory 
and that `worse' measurements in our theory leads to worse performance for
state discrimination tasks. This joins a growing literature using discrimination games to as tasks to which resourceful states out-perform free states in resource theories \cite{Takagi2019,Oszmaniec2019,Carmeli2019,2019skrzypczyk,Uola2019,Skrzypczyk2018}

The origins of resource theories arguably lie in entanglement theory \cite{Horodecki2009, Nielsen2010}. 
As a resource theory, we consider two parties, Alice and Bob, who have access to pure bipartite states and can perform local operations and classical communication (LOCC) for free. As a consequence entangled states become resourceful, as they allow operations which cannot be performed with LOCC and product states alone, such as quantum teleportation \cite{Bennett1993}.

The idea of classical thermodynamics as a resource theory began with the work of Lieb and Yngvason \cite{Lieb1999}. 
They describe macroscopic systems in terms of \emph{adiabatic accessibility}: the states of the system which are accessible from a given state. Without any explicit reference to heat, temperature or thermodynamic cycles, they are able to derive that the possible transformations on a system are characterised by a unique additive function, the entropy.

A simple but illustrative resource theory is the resource theory of non-uniformity \cite{Gour2015}, where maximally mixed states are considered free and the free operations include appending maximally mixed states, global unitaries, and discarding subsystems.
This resource theory classifies quantum states as resourceful the further their spectra are from being uniform, or maximally mixed.

The resource theory structure is very general, and been developed in the abstract \cite{Chitambar2019,2016coecke59,Fritz2015,Brandao2015} as well as well as finding a wide range of applications within quantum mechanics. Further examples include quantum thermodynamics \cite{Goold2016}, reference frames and symmetry \cite{Bartlett2007}, coherence \cite{Streltsov2017}, and knowledge \cite{DelRio2015}.

Since the free operations in the resource theory connect states as a partial order or pre-order, there is a natural composition. Consequently it is perhaps not surprising that the mathematical structure of resource theories was provided in the language of category theory \cite{2016coecke59,Fritz2015} (specifically, \emph{a symmetric monoidal category}), where morphisms from one object to another correspond to free transformations between resource states.

In this paper the `states' or objects of the resource theory will be equivalence classes of positive operator-valued measures (POVMs). The free transformations derive from the ordering on POVMs that results from the ability to make up or confuse measurement outcomes, which cannot improve the ability of a measurement device gain information about the system. That is, our free transformations are a subset of \emph{classical processing} operations on POVMs, sometimes known as \emph{post-processing}. Classical processing on POVMs has been widely studied \cite{1990martens255,Uffink1994,Oszmaniec2017,Oszmaniec2018,Guerini2017,
Kuramochi2015,Dorofeev1997,Heinonen2005,Heinosaari2019,Jencova2008,Buscemi2005,
2011haapasalo1751}, often in the context of measurement compatibility \cite{Carmeli2019,2019skrzypczyk,Uffink1994,Martens1990},  but not in the context of a resource theory with a tensor product structure.
These works often focus on simulating POVMs using projection-valued measurements (or some standard set of POVMs). These theories address the questions of how an experimentalist may simulate a more complicated measurement operation by sampling from various measurement devices which are simpler, or cheaper, to build. In contrast, we are interested in quantifying the ordering of the quality of individual measurement devices.

There is  a growing literature considering quantum measurements from a resource theoretic perspective, similar to our work.
Bischof \emph{et al.} \cite{2018bischof1812.00018} study families of resource theories of coherence defined with respect to general POVMs. Each POVM defines such a family, and can be related to the usual theory of coherence through a Naimark extention of the measurement to a projective measurement. The focus however is on the role of coherence rather than the structure of the set of POVMs.  
Buscemi \emph{et al.} \cite{Buscemi2005} study the ordering relation of classical post-processing to compare it with the ordering relation of pre-processing with a quantum channel.
Skrzypczyk and Linden \cite{Skrzypczyk2018} use resource-theoretic criteria to construct a measure of robustness --- the amount of noise that has to be added before a measurement becomes uninformative, as a quantification of the informativeness of a measurement.
This measure is non-increasing under stochastic combinations of POVM elements, and can be interpreted as the advantage gained in a discrimination game. They also show that it forms a complete set of monotones when taken over all possible discrimination games. They introduce some of the key elements that we independently motivate and formally place in the framework of a resource theory. Skrzypczyk \emph{et al.} \cite{2019skrzypczyk} go further and show that this measure is also important in measurement incompatibility which forms a resource for discrimination tasks.

However, quantum measurements have yet to be given the explicitly formal structure of a resource theory, which we do in this paper. In particular, we consider the tensor product structure, representing performing multiple measurements simultaneously; leading to our results on catalysis and purification.

We further show that many measures of information gain of a quantum measurement, not developed in the context of resource theories, are in fact resource monotones.

Although this resource theory focuses on the ability of measurements to gain information from quantum systems, measurements have many more uses, for example state preparation and quantum control. In addition, quantum measurements are more than just the POVMs; since they cause a back-action on the system. 
Thus there are many possible extensions for a resource theory of measurement: we consider this work as providing some basic structure and results, as well as interpreting previous results such as measures of information gain. This allows for future work to expand the ordering to include a wider variety of uses of quantum measurements.

\section{Partial Order on Measurements}

We begin by defining a partial order on the set of measurements; where by a \emph{measurement} we mean a POVM. That is, an ordered collection of positive semi-definite operators $E_i\ge 0$, one associated with each measurement outcome. The operators satisfy 
the completeness relation
\begin{equation}
\sum_{i=1}^{n} E_{i} = \mathbbm{1}.
\end{equation}
Given a quantum system in state $\rho$, the probability of measuring outcome $i$ is $p(E_{i})=\tr(E_i \rho)$.
We will often write the POVM as a vector of POVM elements,
\begin{equation}
\vec{E} = \left(E_{1},\dots,E_{n}\right),
\end{equation}
in order to succinctly specify transformations of the POVM and to allow the possibility of identical POVM elements; though no formal vector space should be inferred.

We imagine a mischievous gremlin that can modify a measurement device, rewiring the measurement outcomes according to two operations that \emph{a priori} can not improve the device but may deteriorate it (see Fig.~\ref{fig:gremlin}):
\begin{enumerate}
    
    \item \emph{Making up outcomes:} the gremlin can split one outcome into many; that is it can duplicate POVM elements with probabilistic weights, for example:
    \begin{equation}\label{eq:Making_Up_Outcomes}
    \left(E_{1},E_{2},E_{3}\right)\rightarrow \left(E_{1},E_{2},p_{1}E_{3},p_{2}E_{3},p_{3}E_{3}\right),
    \end{equation}
	where $p_{1}+p_{2}+p_{3}=1$.
	More generally, if the gremlin splits each outcome $i$ according to the probability distribution $\vec{P}_{i}$, then the POVM $\vec{E}$ is acted on by a matrix $S$ which is the (matrix) direct sum of these probability vectors
    \begin{equation}\label{eq:op1}
        S = \bigoplus_{i=1}^{n} \vec{P}_{i}.
    \end{equation}
    That is to say, $S$ is a block-diagonal matrix with diagonal blocks that
    consist of one column each.

    \item \emph{Confusing outcomes:} the gremlin can deterministically combine measurement outcomes; and therefore POVM elements. For example:
\begin{equation}
\left(E_{1},E_{2},E_{3}\right)\rightarrow \left(E_{1}+E_{2},E_{3}\right). \label{eq:exop2}
\end{equation}
    That is multiple outcomes may get reported as the same.
    More generally, this transformation on $\vec{E}$ is represented by a
    column stochastic matrix $C$ where each column has a single non-zero entry. These are sometimes called `deterministic matrices'.
\end{enumerate}

We should regard gremlin operations of the first kind as information preserving since they are reversible by operations of the second kind. If $\ell_{i}$ is the length of the probability vector $\vec{P}_{i}$, then we can define the following deterministic matrix in block form
\begin{equation}
R_{S} =
\begin{pmatrix}
R_{1} & \dots & R_{n}
\end{pmatrix}
\end{equation}
where the submatrix $R_{i}$ contains $\ell_{i}$ columns, and contains ones in the $i^{\text{th}}$ row and zeros elsewhere. We see that $R_{S}S = I_{n}$; $R_{S}$ reverses the effect of $S$. The gremlin action of the second kind is typically irreversible, and so we should regard it is informationally destructive. If any POVM elements are proportional, combining them is reversible. In appendix~\ref{app:irreversible} we show that combining POVM elements which are not proportional is irreversible. 

These two operations are both represented by column stochastic matrices. We can show that these two operations are capable of producing \emph{any} column stochastic mixture on a POVM; so the effect of the gremlin is to stochastically scramble the POVM it acts upon. 
\begin{proposition}
Suppose $P$ is an $m\times n$ column stochastic matrix. Then there exists a $mn \times
n$ column stochastic matrix $S$ of the form of \eqref{eq:op1} and a $m \times mn$
column stochastic matrix $C$ whose entries are all either $0$ or $1$ such that $P = CS$.
\end{proposition} 
\begin{proof}
If we denote by $\vec{P}_{i}$ the $i^{\text{th}}$ column of $P$, then define $S$ as
\begin{equation*}
S = \bigoplus_{i=1}^{n} \vec{P}_{i},
\end{equation*}
where the direct sum is the matrix direct sum (as opposed to the vector direct sum). We define $C$ in block form
\begin{equation}
\begin{pmatrix}
C_{1} & \dots & C_{n}
\end{pmatrix},
\end{equation}
where $C_{i} = I_{m}$: the $m\times m$ identity matrix. A simple calculation shows $P=CS$.
\end{proof}

So the gremlin can take the POVM \mbox{$\vec{E} = \left(E_{1},\dots,E_{n}\right)$} and transform it into any POVM \mbox{$\vec{F}=\left(F_{1},\dots,F_{m}\right)$} whose elements have the form
\begin{equation}
F_{i} = \sum_{j=1}^{n} \Pr\left(F_{i}|E_{j}\right)E_{j},
\end{equation}
and $\vec{E}$ should be regarded as at least as good as, if not better than $\vec{F}$ because $\vec{F}$ can be obtained from $\vec{E}$ by making up and confusing outcomes.

Thus we can impose the following order:
\begin{figure}[t]
\includegraphics[scale=0.6]{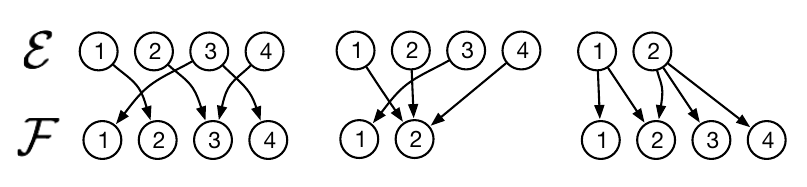}
\caption{Diagram of example free transformations on quantum measurements. The gremlin can make up outcomes and confuse outcomes, or any combination of these two operations. Making up outcomes is reversible; confusing outcomes is only reversible when confusing POVM elements which are proportional to each other. The gremlin can change the number of outcomes of the POVM.}
\label{fig:gremlin}
\end{figure}
\begin{definition}[Order Relation]\label{def:preorder}
For two POVMs, $\vec{E}$ and $\vec{F}$; we say that $\vec{E}\geq \vec{F}$ if $\vec{E}$ can be transformed into $\vec{F}$ via stochastic mixing. That is if for each $F_{i}$,
\begin{equation}
F_{i} = \sum_{j=1}^{n} \Pr\left(F_{i}|E_{j}\right)E_{j}.
\end{equation}
\end{definition}
This order is clearly reflexive: $\vec{E}\geq \vec{E}$ for all POVMs. Since the composition of two column stochastic matrices is also column stochastic, the order is transitive. Thus definition~\ref{def:preorder} specifies a preorder.
Since the gremlin operations in effect implement a stochastic transformation, from here onward, we refer to the `gremlin operations' as `stochastic operations'.

The preorder immediately gives rise to equivalence classes:
\begin{definition}[Equivalence class]\label{def:equivclass}
We say that two POVMs $\vec{E}$ and $\vec{F}$ are \emph{equivalent}, $\vec{E}\sim \vec{F}$ if
\begin{equation}
\vec{F} \leq \vec{E} \quad \text{and}\quad \vec{E} \leq \vec{F}.
\end{equation}
In this case we say that $\vec{E}$ and $\vec{F}$ belong to the same \emph{equivalence class} $\mathcal{E}$.
\end{definition}
We can define the canonical form of any equivalence class, as a POVM without any unnecessary repetitions. This is known as `minimal sufficiency' in \cite{Kuramochi2015}.
\begin{definition}[Canonical Representation]
A \emph{canonical representative} of an equivalence $\mathcal{E}$ is any POVM
$\vec{E}\in \mathcal{E}$ where if $E_{i},E_{j}$ are elements of $\vec{E}$ then
$E_{i}\not\propto E_{j}$ for $i\ne j$.
\end{definition}
Clearly an equivalence class has multiple canonical representatives, all of which are related by permutation matrices. The only equivalence class with a single canonical representative is the class containing the trivial measurement $\left(\mathbbm{1}\right)$.\\ 

The only reversible stochastic operations include making up outcomes, and confusing POVM elements which are proportional to each other. Thus POVMs $\vec{E}$ and $\vec{E^{\prime}}$ are equivalent if each element of $\vec{E}$ is proportional to at least one element of $\vec{E^{\prime}}$ and vice versa. We can hide this unnecessary complexity by considering equivalence classes of POVMs to be the objects of our resource theory.

The preorder arising from definition~\ref{def:preorder} gives rise to a partial order on equivalence classes.

\begin{definition}[Partial Order]\label{def:partialorder}
For equivalence classes $\mathcal{E}$ and $\mathcal{F}$, $\mathcal{E} \succeq \mathcal{F}$ if $\vec{E} \geq \vec{F}$ for any (and thus all) $\vec{E}\in \mathcal{E}$ and $\vec{F}\in\mathcal{F}$.
\end{definition}
In addition to being reflexive and transitive, this order on equivalence classes is anti-symmetric; if $\mathcal{E} \succeq \mathcal{F}$ and $\mathcal{F} \succeq \mathcal{E}$ then $\mathcal{E} = \mathcal{F}$. Hence it is a partial order.

With this partial order we must now identify the \emph{free} resources, i.e. free measurements. 
An equivalence class $\mathcal{I}$ is free if making use of the measurement never provides additional information. 
From a Bayesian perspective this means that our posterior knowledge of the state will always be the same as our a prior knowledge of the state. 
The only measurements that satisfies this are the measurements with elements proportional to the identity operator $\mathbbm{1}$.
These measurements can be implemented by the experimenter ignoring the system and rolling a die to determine the measurement outcome, gaining no information about the system. Hence they are the least resourceful.
These POVMs are all in the same equivalence class with canonical representation $\left(\mathbbm{1}\right)$. 
In this sense the only free resource for this measurement theory of resources is classical probability distributions.

Finally this equivalence class is terminal in the sense that all measurements may be reduced to a free measurement by column stochastic maps. This equivalence class cannot be freely transformed into any other equivalence class.
\begin{proposition}[Free Class is Terminal]\label{prop:TerminalStates}
$\mathcal{E}\succeq \mathcal{I}$ for all $\mathcal{E}$.
\end{proposition}
\begin{proof}
For any POVM we can confuse all elements into the identity $\mathbbm{1}$. So $\vec{E} \geq \left(\mathbbm{1}\right)$, and therefore $\mathcal{E} \succeq \mathcal{I}$. 
\end{proof}

This means that we may dispose of any resource freely, therefore we call this resource theory \emph{waste-free} \cite{2016coecke59}.

We now look at the maximal equivalence classes. These are the the objects that cannot be acquired by freely transforming a more resourceful object.
\begin{definition}[Maximal Objects]
An equivalence class $\mathcal{E}$ is \emph{maximal} if for any $\mathcal{F}$, $\mathcal{F} \succeq \mathcal{E}$ implies $\mathcal{E} = \mathcal{F}$.
\end{definition}

In \cite{Buscemi2005}, the authors call maximal objects, `clean'. They also prove the following proposition.

\begin{proposition}
An equivalence class $\mathcal{E}$ is maximal if and only if all POVMs contain rank-1 elements.
\end{proposition}
\begin{proof}
Assume $\vec{E}$ consists of rank-1 elements and $\vec{E} \leq \vec{F}$. We have for all $i\in\left\lbrace 1,\dots,n\right\rbrace$
\begin{equation}
E_{i} = \sum_{j=1}^{m} \Pr\left(E_{i}|F_{j}\right)F_{j}.
\end{equation}
But this can only hold if $F_{j}\propto E_{i}$ for all pairs of $i$ and $j$ where $\Pr\left(E_{i}|F_{j}\right)$ is nonzero. If this is not the case then the rank of $E_{i}$ must be greater than $1$. Hence $\vec{E}$ and $\vec{F}$ are equivalent $\vec{E} \sim \vec{F}$, and $\vec{E}$ is maximal.

Conversely, consider the spectra of the elements of and arbitrary POVM $\vec{E} = \left(E_{1},\dots,E_{n}\right)$,
\begin{equation}
E_{i} = \sum_{j=1}^{d} \lambda_{j}^{i}\ketbra{\lambda_{j}^{i}}{\lambda_{j}^{i}}.
\end{equation}
Now consider the POVM
\begin{equation}
\vec{E^{\prime}} = \left(\lambda_{1}^{1}\ketbra{\lambda_{1}^{1}}{\lambda_{1}^{1}},\dots,\lambda_{d}^{n}\ketbra{\lambda_{d}^{n}}{\lambda_{d}^{n}}\right).
\end{equation}
Clearly $\vec{E}\leq \vec{E^{\prime}}$ since $E_{i}$ is the result of combining all elements of the form $\lambda_{j}^{i}\ketbra{\lambda_{j}^{i}}{\lambda_{j}^{i}}$, and the elements of $\vec{E^{\prime}}$ are all rank-1. So every POVM $\vec{E}$ is the result of a stochastic operation upon a POVM which consists only of rank-1 elements. Furthermore, we showed in appendix~\ref{app:irreversible} that combining POVM elements is reversible if and only if those elements are proportional to each other, that is, if and only if $\vec{E}$ already consists of only rank-1 elements. Therefore, if $\vec{E}$ does not consist of rank-1 elements then $\vec{E^{\prime}}$ is a member of a separate equivalence class and $\vec{E}$ is not maximal.
\end{proof}

\section{Measurement Resource Theory}

The mathematical structure of a resource theories as symmetric monoidal categories was provided by Coecke \emph{et al.} \cite{2016coecke59,Fritz2015} in the language of category theory.
As mentioned in the introduction, the \emph{objects} are the equivalence classes of POVMs and the morphisms are derived from the ordering described in definition~\eqref{def:partialorder}.
That is, there is a \emph{morphism} or \emph{free transformation} from equivalence class $\mathcal{E}$ to $\mathcal{F}$ if $\mathcal{E}\succeq \mathcal{F}$.

Since for all equivalence classes we have $\mathcal{E}\succeq \mathcal{E}$, there exists an morphism from all objects to themselves, which we identify as the \emph{identity morphism} for each object, a requirement of all categories.
The \emph{composition of morphisms} is automatically defined due to the transitive property of the partial order resulting from the composition of column stochastic maps.
These properties make the set of equivalence classes and free transformations a category, which automatically follows from the partial order (any partially ordered set can be formulated as a category).
All that remains is to define the tensor product between equivalence classes to make this category symmetric and monoidal.

We wish to utilise the standard tensor product on linear operators, but generalised to equivalence classes. 
Let $\vec{E}\in\mathcal{E}$ and $\vec{F}\in\mathcal{F}$ be POVMs on Hilbert spaces $\mathcal{H}_{A}$ and $\mathcal{H}_{B}$ respectively. We define the object $\mathcal{E}\otimes\mathcal{F}$ on $\mathcal{H}_{A}\otimes \mathcal{H}_{B}$ as the equivalence class of $\vec{E}\otimes\vec{F}$,
\begin{equation}
\mathcal{E}\otimes\mathcal{F} = \set{\vec{A}\; |\; \vec{E}\otimes\vec{F}\sim \vec{A}}.\label{eq:productmeas}
\end{equation}
Not all POVMs within the equivalence class of $\vec{E}\otimes\vec{F}$ can be written as the tensor product of two POVMs. Similarly, not all equivalence classes of POVMs on combined systems can be expressed as the tensor product of two equivalence classes; an example is the equivalence class containing a Bell measurement on two qubits. Furthermore, equivalence classes which can't be expressed as a tensor product can be reached by the stochastic operation on a tensor product of two classes. Suppose $\vec{A}\otimes\vec{B}$ is a tensor product of POVMs on Hilbert space $\mathcal{H}_{A}\otimes\mathcal{H}_{B}$. We have, for example, implement the transformation
\begin{equation}
\vec{A}\otimes\vec{B}=
\begin{pmatrix}
A_{1}\otimes B_{1} \\
A_{1}\otimes B_{2} \\
A_{2}\otimes B_{1} \\
A_{2}\otimes B_{2}
\end{pmatrix} \rightarrow
\begin{pmatrix}
A_{1}\otimes B_{1} + A_{2}\otimes B_{2} \\
A_{2}\otimes B_{1} + A_{1}\otimes B_{2}
\end{pmatrix},
\end{equation}
and the latter POVM cannot be written as the tensor product of two POVMs on $\mathcal{H}_{A}\otimes\mathcal{H}_{B}$. Nevertheless these equivalence classes are valid measurements and therefore are objects in the resource theory. Indeed there are POVMs within the equivalence class of a tensor product of POVMs that cannot be written as the tensor product of two local POVMs. For example, we can freely (and reversibly) append a $0$ operator to $\vec{A}\otimes \vec{B}$; but after this transformation it cannot be written as the tensor product of two POVMs.  
In order to decide if an equivalence class can be written as a product of two equivalence classes, we define product measurements, analogous to product states of quantum systems.

\begin{definition}[Product measurement]
An equivalence class $\mathcal{A}$ is a \emph{product measurement} on $\mathcal{H}_{A}\otimes\mathcal{H}_{B}$ if there exists a canonical representative of the form $\vec{E}\otimes \vec{F}$, where $\vec{E}\in\mathcal{E}$ is defined on $\mathcal{H}_{A}$ and $\vec{F}\in\mathcal{F}$ is defined on $\mathcal{H}_{B}$. We can then write $\mathcal{A}=\mathcal{E}\otimes\mathcal{F}$.
\end{definition}

The tensor product of equivalence classes provides the \emph{monoidal product} of the category. This leaves defining the unit object (our free measurements) which we may freely append to and discard from any measurement. This is the equivalence class of free POVMs that we defined in proposition \ref{prop:TerminalStates}, $\mathcal{I}$, corresponding to measurements which provide the experimenter with no information about any system that could potentially have been measured. This is for all equivalence classes $\mathcal{A}$,
\begin{equation}
    \mathcal{A}\otimes \mathcal{I} \sim \mathcal{A} \sim \mathcal{I}\otimes \mathcal{A}.
\end{equation}
This provides all the properties needed for our category to be a \emph{monoidal category}.

Finally note that $\mathcal{E}\otimes \mathcal{F}$ and $\mathcal{F}\otimes \mathcal{E}$ are equivalent, since it does not matter which equivalence class is considered `first' or `second'. This equivalence providing the final element of the mathematical structure of a resource theory, the \emph{symmetric} property of the monoidal category, leaving us with a \emph{symmetric monoidal category}.

Having discussed the concept of a product measurement, it is worthwhile defining a reverse procedure: a reduced measurement, analogous to the reduced density matrix.
\begin{definition}[Reduced Measurement]
Suppose $\vec{E}=\left(E_{1},\dots,E_{n}\right)$ is a POVM on Hilbert space $\mathcal{H}_{A}\otimes\mathcal{H}_{B}$. Then the reduced measurement on Hilbert space $\mathcal{H}_{A}$ is
\begin{equation}
\vec{E}_{A} = \frac{1}{d_{B}}\left(\ptr{B}{E_{1}},\dots,\ptr{B}{E_{n}}\right),
\end{equation}
where $d_{B} = \dim\left(\mathcal{H}_{B}\right)$.
\end{definition}
Clearly, $\vec{E}_{A}$ is a valid POVM, since it contains positive semi-definite operators which sum to the identity. However the processes that produces reduced measurements are \emph{not} free transformations. That being said, they are a useful mathematical tool to understand the structure of this resource theory.

The process of reducing measurements is well defined on equivalence classes. That is, if $\vec{E} \sim \vec{E}^{\prime}$ then $\vec{E}_{A} \sim \vec{E}^{\prime}_{A}$. This results from the linearity of the partial trace. Since $\vec{E}^{\prime}$ will have elements of the form $E^{\prime}_{i}\propto E_{k}$ for some element $E_{k}$ in $\vec{E}$, then the reduced measurement $\vec{E}^{\prime}_{A}$ will have elements $E^{\prime}_{A,\,i} \propto \ptr{B}{E_{k}}$. This is within the same equivalence class as $\vec{E}_{A}$ which has elements of the form $\ptr{B}{E_{k}}$, hence the reduced measurement of two POVMs of the same equivalence class stay equivalent. 

It is therefore meaningful to reduce equivalence classes themselves, by reducing all the elements within each equivalence class.
It is easy to see that if $\mathcal{E}\otimes\mathcal{F}$ is a product measurement on Hilbert space $\mathcal{H}_{A}$ and $\mathcal{H}_{B}$, then the reduced measurement on Hilbert space $\mathcal{H}_{A}$ is $\mathcal{E}$.

Having defined the category we can consider the notions of \emph{catalysis} and \emph{purification} in the POVM resource theory of measurements. These are situations where extra resourceful states enable transformations which cannot be freely performed.
\begin{definition}[Catalysis]
An equivalence class $\mathcal{C}$ is a \emph{catalyst} for the transformation $A\succeq B$, if 
\begin{equation}
\mathcal{A} \otimes \mathcal{C} \succeq \mathcal{B}\otimes \mathcal{C} \quad \text{but}\quad \mathcal{A} \nsucceq \mathcal{B}.
\end{equation}
\end{definition}
That is, the resource $\mathcal{C}$ is not consumed, but allows $A$ to be freely transformed into $B$, which it otherwise could not be.
If a resource theory contains no catalysts then it is known naturally enough as \emph{catalysis-free}.
We may also define purification of a resource in a similar manner.
\begin{definition}[Purification]
An object $\mathcal{A}$ may be a \emph{purified} into a more resourceful object $\mathcal{B} \succeq \mathcal{A}$, where $\mathcal{B}\ne \mathcal{A}$, if 
\begin{equation}
\mathcal{A}^{\otimes n} \succeq \mathcal{B}^{\otimes n^{\prime}}\otimes \mathcal{S} .
\end{equation}
In which case it is said to have rate $\frac{n}{n^{\prime}}$.
\end{definition}
So enough copies of $\mathcal{A}$ might allow it to be freely transformed back into some number of copies of $\mathcal{B}$.

We can show however that this resource theory of quantum measurement does not contain either of these two features. They are corollaries of the following theorem.

\begin{theorem}\label{thm:ReducingStochasticMaps}
Suppose $\mathcal{A}\otimes \mathcal{B}$ and $\mathcal{C}\otimes \mathcal{D}$ are product measurements on the Hilbert space $\mathcal{H}_{A}\otimes \mathcal{H}_{B}$ and $\mathcal{A}$ and $\mathcal{C}$ are defined on $\mathcal{H}_{A}$. If $\mathcal{A}\otimes \mathcal{B}\succeq \mathcal{C}\otimes \mathcal{D}$ then $\mathcal{A}\succeq \mathcal{C}$.
\end{theorem}
\begin{proof}
Let $\vec{A}\otimes \vec{B}$ and $\vec{C}\otimes \vec{D}$ be the canonical representatives of $\mathcal{A}\otimes \mathcal{B}$ and $\mathcal{C}\otimes\mathcal{D}$ respectively and suppose that $\vec{A}\otimes \vec{B}\geq \vec{C}\otimes \vec{D}$. This means that each element of $\vec{C}\otimes \vec{D}$ has the form
\begin{equation}
C_{\mu} \otimes D_{\nu} = \sum_{i,j} \Pr(C_{\mu} \otimes D_{\nu}|A_{i}\otimes B_{j}) A_{i}\otimes B_{j}.
\end{equation}
Then the reduced measurement on Hilbert space $\mathcal{H}_{A}$ is equivalent to $\vec{C}$, with elements of the form
\begin{equation}
C_{\mu}\frac{\tr(D_{\nu})}{d_{B}} = \sum_{i,j} \Pr(C_{\mu} \otimes D_{\nu}|A_{i}\otimes B_{j}) \frac{\tr(B_{j})}{d_{B}} A_{i}
\end{equation}
However this is a stochastic mixture of $\vec{A}$, since 
\begin{equation}
\sum_{\mu,\nu} \sum_{j} \Pr(C_{\mu} \otimes D_{\nu}|A_{i}\otimes B_{j}) \frac{\tr(B_{j})}{d_{B}} = 1 \quad \text{for all }\; i.
\end{equation}
\end{proof}
This theorem immediately implies that this resource theory doesn't contain catalysis or purification.
\begin{corollary}[No Catalysis]\label{lem:NoCatalysis}
Suppose $\mathcal{E}\otimes \mathcal{F} \succeq \mathcal{E}^{\prime}\otimes\mathcal{F}$, then $\mathcal{E}\succeq \mathcal{E}^{\prime}$.
\end{corollary}
\begin{proof}
This follows from theorem~\ref{thm:ReducingStochasticMaps} in the case $\mathcal{B} = \mathcal{D}$.
\end{proof}
\noindent Thus the POVM resource theory of measurement is said to be catalysis-free.

\begin{corollary}[No Purification]\label{lem:NoPurification}
Suppose $\mathcal{C} \succeq \mathcal{A}$ and
\begin{equation}
\mathcal{A}^{\otimes n} \succeq \mathcal{C}^{\otimes n^{\prime}}\otimes \mathcal{S} ,
\end{equation}
then $\mathcal{A}\succeq \mathcal{C}$; and therefore $\mathcal{A} = \mathcal{C}$.
\end{corollary}
\begin{proof}
This follows from theorem~\ref{thm:ReducingStochasticMaps} using the substitution $\mathcal{B} = \mathcal{A}^{\otimes n-1}$ and $\mathcal{D} = \mathcal{C}^{\otimes n^{\prime}-1}\otimes \mathcal{S}$.
\end{proof}

\section{Resource Monotones}\label{sec:mrtmonotones}

Onto this structure we can define resource monotones and show that a number of standard measurements of information gain are valid monotones on this resource theory.

\begin{definition}[Resource Monotone]\label{def:monotone}
A \emph{resource monotone} is a function $\mu$ from the objects to the real numbers, whose order respects that of the resource theory:
\begin{equation}
\mathcal{A} \succeq \mathcal{B} \quad \Rightarrow\quad \mu\left(\mathcal{A}\right)\geq \mu\left(\mathcal{B}\right).
\end{equation}
\end{definition}
Monotones can quantify how much more resourceful one object is over another. They also indicate transformations that are impossible.

A monotone on equivalence classes will naturally follow from a monotone $\mu$ on POVMs which satisfies the condition that for any two POVMs $\vec{E} = \left(E_{1},\dots,E_{n}\right)$ and $\vec{F} = \left(F_{1},\dots,F_{m}\right)$, if $\vec{F}$ is a stochastic mixture of $\vec{E}$,
\begin{align}
F_{i} &= \sum_{k=1}^{n} \Pr\left(F_{i}|E_{j}\right)E_{j} \quad \text{for all }i\in\set{1,\dots,m}, \nonumber \\
&\text{then}\quad \mu\left(\vec{E}\right)\geq \mu\left(\vec{F}\right). \label{eq:convexmix}
\end{align}
If this condition is satisfied, then it automatically follows from definition~\eqref{def:equivclass} that $\mu$ will be constant on equivalence classes. Hence we can define $\mu\left(\mathcal{E}\right) := \mu\left(\vec{E}\right)$ where $\vec{E}$ is any element of $\mathcal{E}$, and therefore $\mu\left(\mathcal{E}\right) \geq \mu\left(\mathcal{F}\right)$ if $\vec{E} \in\mathcal{E} \geq \vec{F}\in\mathcal{F}$.

A single monotone usually cannot completely classify the possible transformations of a resource theory, as resource theories are typically only partially ordered or preordered. However this is possible with a family of monotones.
\begin{definition}
A set of monotones $\left\lbrace \mu_{i} | i\in I\right\rbrace$ is a \emph{complete family of resource monotones} if
\begin{equation}
\mu_{i}\left(\mathcal{A}\right)\geq \mu_{i}\left(\mathcal{B}\right) \; \text{for all }i \quad \Rightarrow \quad \mathcal{A} \succeq \mathcal{B}.
\end{equation}
\end{definition}
Coecke \emph{et al.} \cite{2016coecke59} show it is always possible to find a complete family of monotones, although the set may be uncountably infinite. Skrzypczyk and Linden \cite{Skrzypczyk2018} showed that state discrimination games form a complete family of monotones (see section~\ref{sec:qsd}), however this is also an uncountably infinite set.

Here we wish to establish the existence of some resource monotones with respect to the resource theory of quantum measurement. 
Since the free transformations can be irreversible, they can lose information. We might then expect some standard measures of the information gained from a POVM \cite{Winter2004,MWilde2012,Maccone2007,08buscemi210504,Banaszek,Skrzypczyk2018} to be resource monotones. 
We will show that several different types of measures of information gain are monotones. Maccone \cite{Maccone2007} and Buscemi \emph{et al.} \cite{08buscemi210504} define information gain using the Shannon and von-Neuamann entropy functions respectively. In fact the Buscemi measure of information gain is argued to be `\emph{the} information-theoretic measure of information gain of a quantum measurement' \cite{Berta2014}. Banaszek \cite{Banaszek} defines information gain using the fidelity function. We also find that the `robustness of measurement' introduced by Skrzpczyk and Linden \cite{Skrzypczyk2018} is also a monotone on the POVM resource theory of measurement, which is defined in terms of the operator norm. 

Maccone \cite{Maccone2007} defined a measure of information gain for the POVM $\vec{E}=\left(E_{1},\dots,E_{n}\right)$ as the mutual information between the eigenstates of state $\rho$ being measured, and the measurement outcomes. If $\rho$ has the spectral decomposition
\begin{equation*}
\rho = \sum_{k=1}^{d} \lambda_{k}\ketbra{\psi_{k}}{\psi_{k}},
\end{equation*}
where $d$ is the dimension of the quantum system. The Maccone measure of information gain $I_{\text{Mac}}$ can be written in terms of the classical Shannon entropy $H$,
\begin{equation}
I_{\text{Mac}}\left(\rho,\vec{E}\right) = H\left(\lambda\left(\rho\right)\right) - \sum_{i=1}^{m} \Pr\left(E_{i}\right)H\left(\vec{q}_{i}^{\vec{E}}\right).
\end{equation}
where $\lambda\left(\rho\right)$ is the probability vector containing the eigenvalue spectrum of $\rho$, $\Pr\left(E_{i}\right) = \tr{\left(\rho E_{i}\right)}$, and
\begin{equation}
\vec{q}_{i}^{\vec{E}} = \left(\frac{\lambda_{1}\bra{\psi_{1}}E_{i}\ket{\psi_{1}}}{\tr(E_{i}\rho)},\dots,\frac{\lambda_{d}\bra{\psi_{d}}E_{i}\ket{\psi_{d}}}{\tr(E_{i}\rho)}\right).
\end{equation}

We can show that this function decreases under stochastic operations, by invoking \eqref{eq:convexmix}. For a single element of $\vec{q}_{i}^{\vec{F}}$,
\begin{equation}
\begin{split}
\frac{\lambda_{k}\bra{\psi_{k}}F_{i}\ket{\psi_{k}}}{\tr{\left(F_{i}\rho\right)}} &= \sum_{j=1}^{n} \frac{\Pr\left(F_{i}|E_{j}\right)}{\tr(F_{i}\rho)}\lambda_{k}\bra{\psi_{k}}E_{j}\ket{\psi_{k}} \\
&= \sum_{j=1}^{n} \frac{\Pr\left(E_{j}|F_{i}\right)}{\tr(E_{j}\rho)}\lambda_{k}\bra{\psi_{k}}E_{j}\ket{\psi_{k}},
\end{split}
\end{equation}
where we used Bayes' theorem in the last equality. Hence we can write
\begin{equation}
\vec{q}_{i}^{\vec{F}} = \sum_{j=1}^{n} \Pr\left(E_{j}|F_{i}\right) \vec{q}_{j}^{\vec{E}}.
\end{equation}
The monotonicity of this measure then follows from the concavity of the Shannon entropy,
\begin{align}
&H\left(\lambda\left(\rho\right)\right) - \sum_{i=1}^{m} \Pr\left(F_{i}\right)H\left(\vec{q}_{i}^{\vec{F}}\right)\nonumber \\
&\leq H\left(\lambda\left(\rho\right)\right) - \sum_{i=1}^{m}  \sum_{j=1}^{n} \Pr\left(F_{i}\right)\Pr\left(E_{j}|F_{i}\right)H\left(\vec{q}_{j}^{\vec{E}}\right)\nonumber \\
&=H\left(\lambda\left(\rho\right)\right) - \sum_{j=1}^{n} \Pr\left(E_{j}\right)H\left(\vec{q}_{j}^{\vec{E}}\right)\nonumber \\
&=I_{\text{Mac}}\left(\rho,\vec{E}\right).
\end{align}

A similar argument also holds for the measure of information gain proposed by Buscemi \emph{et al.} \cite{08buscemi210504}, again due to the concavity of the von Neumann entropy (see appendix~\ref{app:buscemi}).

In addition to the above entropic measures of information gain we also find that the Banaszek \cite{Banaszek} measure of information gain, is also a resource monotone.
The Banaszek measure of information gain is written in terms of the fidelity function
\begin{equation}
I_{\text{Ban}}(\vec{E}) = \frac{1}{d\left(d+1\right)}\left(d+\sum_{i=1}^{m}\bra{\psi_{i}^{\vec{E}}}E_{i}\ket{\psi_{i}^{\vec{E}}}\right),
\end{equation}
where $\ket{\psi_{i}^{\vec{E}}}$ is the state which maximises $\bra{\psi_{i}^{\vec{E}}}E_{i}\ket{\psi_{i}^{\vec{E}}}$. It is easy to show that this measure is a resource monotone, since
\begin{align}
\frac{1}{d\left(d+1\right)}&\left(d+\sum_{i=1}^{m}\sum_{j=1}^{n}\Pr\left(F_{i}|E_{j}\right)\bra{\psi_{i}^{\vec{F}}}E_{j}\ket{\psi_{i}^{\vec{F}}}\right) \nonumber \\
&\leq \frac{1}{d\left(d+1\right)}\left(d+\sum_{j=1}^{m}\bra{\psi_{j}^{\vec{E}}}E_{j}\ket{\psi_{j}^{\vec{E}}}\right),
\end{align}
since $\ket{\psi_{j}^{\vec{E}}}$ maximises $\bra{\psi_{j}^{\vec{E}}}E_{j}\ket{\psi_{j}^{\vec{E}}}$.

Recent work by Skrzypczyk \emph{et al.} \cite{Skrzypczyk2018} has shown that a measure called the `robustness of measurement' is a monotone for a POVMs. This is defined as,
\begin{equation}
I_{\text{Skr}}(\vec{E}) = \sum_{i=1}^{n} || E_{i} || - 1,
\end{equation}
and due to its monotonic nature under stochastic mixing it is a resource monotone for the POVM resource theory of measurements.

This monotonicity arises from the properties of the operator norm; specifically
\begin{equation}
|| a E || = |a|\,|| E ||, \quad ||E+F|| \leq ||E|| + ||F||.
\end{equation}
Thus
\begin{align}
\sum_{i=1}^{m} ||F_{i}|| &= \sum_{i=1}^{m} || \sum_{j=1}^{n} \Pr\left(F_{i}|E_{j}\right) E_{j} || \nonumber
\\
&\leq \sum_{i=1}^{m}\sum_{j=1}^{n} || \Pr\left(F_{i}|E_{j}\right) E_{j} || \nonumber \\ 
&= \sum_{i=1}^{m}\sum_{j=1}^{n} \Pr\left(F_{i}|E_{j}\right) || E_{j} || = \sum_{j=1}^{n} || E_{j} ||.
\end{align}
So the sum of the operator norms of the POVM elements is a monotone and it is bounded from above by $d$, the Hilbert space dimension, which is saturated only in the case of maximal POVMs; those whose elements consist only of rank-1 operators. It is bounded from below by $1$, which is the norm of the trivial POVM $\left(\mathbbm{1}\right)$. Therefore functions based off the operator norm may give rise to a number of further resource monotones that could be constructed.

Finally, it is common in resource theories to derive a majorisation relation which describes the transformation. The state $\rho$ majorises the state $\sigma$ if
\begin{equation}
\sum_{i=1}^{k} \lambda^{\downarrow}_{i}\left(\rho\right) \geq \sum_{i=1}^{k} \lambda^{\downarrow}_{i}\left(\sigma\right) \, \text{ for all }k\in\set{1,\dots,n}, \label{eq:majorisation}
\end{equation}
where $\lambda^{\downarrow}\left(\rho\right)$ is the vector of eigenvalues of $\rho$ written in non-decreasing order. 
For example, in the resource theory of entanglement, the bipartite pure state $\ket{\psi}$ can be transformed into $\ket{\phi}$ using local operations and classical communication if and only if
\begin{equation}
\ptr{B}{\ketbra{\phi}{\phi}} \succeq \ptr{B}{\ketbra{\psi}{\psi}}
\end{equation}
where we have traced over one subsystem $B$. Majorisation relations are especially useful since their definition \eqref{eq:majorisation} provides a finite, complete set of monotones. We can find a majorisation condition our free transformations on POVMs, however it does not fully characterise the free transformations.
\begin{proposition}
Let $\vec{E}=\left(E_{1},\dots,E_{n}\right)$ and $\vec{F}=\left(F_{1},\dots,F_{m}\right)$ be POVMs.
\begin{equation}
\vec{E} \geq \vec{F}\; \Rightarrow\; \sum_{i=1}^{n} \lambda^{\downarrow}\left(E_{i}\right) \succeq \sum_{i=1}^{m} \lambda^{\downarrow}\left(F_{i}\right).
\end{equation}
\end{proposition}
\begin{proof}
The proof relies on the result for Hermitian operators wherein $\lambda^{\downarrow}\left(A\right)+\lambda^{\downarrow}\left(B\right) \succeq \lambda^{\downarrow}\left(A+B\right)$. 

Since $\vec{E} \geq \vec{F}$, they are related by a column stochastic matrix \eqref{eq:convexmix}. Thus we have
\begin{equation}
\sum_{i=1}^{m} \lambda^{\downarrow}\left(F_{i}\right) = \sum_{i=1}^{m} \lambda^{\downarrow}\left(\sum_{j=1}^{n}\Pr\left(F_{i}|E_{j}\right)E_{j}\right).
\end{equation}
From the definition of majorisation \eqref{eq:majorisation} it is easy to see that if $\lambda^{\downarrow}\left(\rho_{1}\right)\succeq \lambda^{\downarrow}\left(\rho_{2}\right)$ and $\lambda^{\downarrow}\left(\sigma_{1}\right)\succeq \lambda^{\downarrow}\left(\sigma_{2}\right)$ then $\lambda^{\downarrow}\left(\rho_{1}\right)+\lambda^{\downarrow}\left(\sigma_{1}\right)\succeq \lambda^{\downarrow}\left(\rho_{2}\right)+\lambda^{\downarrow}\left(\sigma_{2}\right)$. So we have
\begin{align}
\sum_{i=1}^{m} \lambda^{\downarrow}\left(\sum_{j=1}^{n}\Pr\left(F_{i}|E_{j}\right)E_{j}\right) &\preceq \sum_{i=1}^{m} \sum_{j=1}^{n} \lambda^{\downarrow}\left(\Pr\left(F_{i}|E_{j}\right)E_{j}\right)\nonumber \\
&= \sum_{i=1}^{m} \sum_{j=1}^{n} \Pr\left(F_{i}|E_{j}\right) \lambda^{\downarrow}\left(E_{j}\right)\nonumber \\
&= \sum_{j=1}^{n} \lambda^{\downarrow}\left(E_{j}\right).
\end{align}
\end{proof}

It is easy to show that the converse doesn't hold. For if we have $\vec{E} = \left(\ketbra{0}{0},\ketbra{1}{1}\right)$ and $\vec{F} = \left(\ketbra{+}{+},\ketbra{-}{-}\right)$, then $\lambda^{\downarrow}\left(\ketbra{0}{0}\right)+\lambda^{\downarrow}\left(\ketbra{1}{1}\right) \succeq \lambda^{\downarrow}\left(\ketbra{+}{+}\right)+\lambda^{\downarrow}\left(\ketbra{-}{-}\right)$, but clearly $\vec{E}$ cannot be freely transformed into $\vec{F}$.

Hence this majorisation condition does not fully characterise all possible free transformations.

\section{Application: Quantum State Discrimination}\label{sec:qsd}

The task of quantum state discrimination requires using quantum measurements to acquire information from a system, so a less resourceful measurement should be worse in its ability to perform this task. In fact, a there is a growing literature using discrimination games as tasks to distinguish free states from resourceful states, in more general resource theories \cite{Takagi2019,Oszmaniec2019,Carmeli2019,2019skrzypczyk,Uola2019}.

In this setting, Alice sends Bob a quantum state from a predefined alphabet $\set{\rho_{1},\dots,\rho_{k}}$ with prior probabilities $\Pr\left(\rho_{i}\right)$. Bob then seeks the best measurement with which to determine the state sent by Alice with the minimal error (or maximal success); that is he wants to gain the most amount of information to decide which state was sent. In the typical setting, this means finding the POVM which maximises the average probability of Bob's measurement reporting the sent state by Alice
\begin{equation}
\sum_{i=1}^{n} \Pr\left(A_{i}|\rho_{i}\right)\Pr\left(\rho_{i}\right).\label{eq:oldsa}
\end{equation}

If the measurement contains more outcomes than there are states in Alice's alphabet, Bob will decide on the state with the highest posterior probability. Bob should therefore find the POVM which maximises
\begin{equation}
\expect{s_{\vec{A}}} := \sum_{i=1}^{m} \Pr\left(\rho^{\vec{A}}_{i^{*}}|A_{i}\right)\Pr\left(A_{i}\right), \label{eq:sa}
\end{equation}
where $\rho_{i^{*}}^{\vec{A}}$ maximises the posterior distribution of the alphabet states upon receiving the measurement outcome $A_{i}$,
\begin{equation}
\rho_{i^{*}}^{\vec{A}} := \underset{\rho\in \set{\rho_{1},\dots,\rho_{k}}}{\argmax} \; \Pr\left(\rho|A_{i}\right).\label{eq:maxrho}
\end{equation}
With this framework, it is clear that the free transformations reduce Bob's ability to distinguish between the states sent by Alice. 
Once again supposing that $\vec{F}$ is a stochastic mixture of $\vec{E}$, we have
\begin{align}\label{eq:StateDescrimination_Monotone}
\expect{s_{\vec{F}}} &= \sum_{i=1}^{m} \Pr\left(\rho^{\vec{F}}_{i^{*}}|F_{i}\right)\Pr\left(F_{i}\right) = \sum_{i=1}^{m} \Pr\left(F_{i}|\rho^{\vec{F}}_{i^{*}}\right)\Pr\left(\rho^{\vec{F}}_{i^{*}}\right) \nonumber \\
&= \sum_{i=1}^{m} \sum_{j=1}^{n} \Pr\left(F_{i}|E_{j}\right) \Pr\left(E_{j}|\rho^{\vec{F}}_{i^{*}}\right)\Pr\left(\rho^{\vec{F}}_{i^{*}}\right) \nonumber \\
&= \sum_{i=1}^{m} \sum_{j=1}^{n} \Pr\left(F_{i}|E_{j}\right) \Pr\left(\rho^{\vec{F}}_{i^{*}}|E_{j}\right)\Pr\left(E_{j}\right) \nonumber \\
&\leq \sum_{i=1}^{m} \sum_{j=1}^{n} \Pr\left(F_{i}|E_{j}\right) \Pr\left(\rho^{\vec{E}}_{j^{*}}|E_{j}\right)\Pr\left(E_{j}\right) \nonumber \\
&=\sum_{j=1}^{n} \Pr\left(\rho^{\vec{E}}_{j^{*}}|E_{j}\right)\Pr\left(E_{j}\right) = \expect{s_{\vec{E}}}.
\end{align}
We note that using the canonical probability of success \eqref{eq:oldsa}, it would be possible to improve Bob's ability to discriminate between states in Alice's alphabet using a free transformation.

Skrzypczyk and Linden \cite{Skrzypczyk2018} showed that all state discrimination games constitute a complete family of monotones. That is, if $\expect{s_{\vec{E}}}\geq \expect{s_{\vec{F}}}$ \eqref{eq:sa} for all possible alphabets and prior probabilities, then $\vec{E} \geq \vec{F}$. For completeness we reproduce this proof in appendix~\ref{app:games}. The number of possible discrimination games is uncountably infinite, and it is still an open question whether there exists a \emph{finite} complete family of monotones which characterise all transformations between POVMs. 

\section{Conclusion}

The primary task of a quantum measurement is to gain information about a quantum system. With respect to that task, we would intuitively expect rank-1 projective measurements should be considered `maximally resourceful', while simply rolling a die gains no information about the system at all and so is `least resourceful'. 
We have formalised this intuition into a resource theory, by considering the free operations of making up measurement outcomes and confusing outcomes, neither of which improve the ability of a measurement to gain information about the system.
With these two free operations any column stochastic mixture can be
performed on the elements of a POVM. 
The order arising from these free operations automatically gives rise to the basic category structure of a resource theory, leaving only the tensor product and unit object (free resource) structure remaining to be defined.

By devising a formal resource theory for quantum measurements, we can make two observations. First, we have found that catalysis and purification, which are key protocols in entanglement theory, are not possible for quantum measurements under the free operations of this resource theory. Such limitations on measurements may have intriguing implications for other areas of quantum information theory. Second, we have shown that previous proposals for measures of information gain are indeed resource monotones as would be expected and they quantify how much more resourceful one POVM is over another. However, resource theory for POVMs is \emph{not a total order}. This means that a single monotone cannot capture the relationships between measurements and so there cannot exist a single notion of information gain. Acquiring information in quantum mechanics using a measurement is more complex than can be captured by a single measure.

As a practical example, we showed that a reduced ability to gain information from a system implies a reduced ability to discriminate between states in an ensemble. That is, the probability of success in discriminating between quantum states in an ensemble is a monotone for the resource theory of quantum measurements.

This resource has focused on the ability of a quantum measurement to gain information about a system. Of course, the corollary effect of a quantum measurement is the disturbance by the measurement back-action \cite{Maccone2007,08buscemi210504}. It would be an exciting avenue for future work to include the disturbance into the resource theory of quantum measurement \cite{Moreira2019}. For example, two quantum measurements could be equally resourceful from the perspective of gaining information, but one might cause more disturbance to the system than the other, rendering it less valuable.

\begin{acknowledgments}
This research was funded in part by the Australian Research Council Centre of Excellence for Engineered Quantum Systems (Project number CE170100009). 
\end{acknowledgments}

\bibliography{references}

\appendix
\section{Proof: Combining non-proportional POVMs is irreversible}\label{app:irreversible}

Confusing elements in any finite POVM is reversible if and only if the elements are proportional. So the effect of any deterministic matrix can only be undone when it combines two proportional elements. This need only be shown for the case of confusing two elements, since any combination of multiple elements can be decomposed into successive combinations of just two elements.

\begin{proposition}
Suppose $\left( E_{1}+E_{2},E_{3},\dots,E_{n}\right)$ is a POVM, and suppose there exists stochastic coefficients $a_{i,\,j}$ such that for all $i\in\set{1,\dots,n}$,
\begin{equation}
E_{i} = a_{i,\,1}\left(E_{1}+E_{2} \right)+\sum_{j=3}^{n} a_{i,\,j}E_{j}, \label{eq:mix}
\end{equation}
where $a_{i,\,j}\geq 0$ and $\sum_{i=1}^{n} a_{i,\,j}=1$ for all $j$; then $E_{1}\propto E_{2}$. 
\end{proposition}
\begin{proof}
In the case $n=2$, we have
\begin{align*}
a_{1,\,1}\left(E_{1}+E_{2}\right) &= E_{1}, \\
a_{2,\,1}\left(E_{1}+E_{2}\right) &= E_{2},
\end{align*}
which implies $E_{1}\propto E_{2}$.
We now show that in general, \eqref{eq:mix} always leads to a requirement such as this. More precisely we show that \eqref{eq:mix} implies that $\left(E_{1},E_{2},\dots,E_{n-1}\right)$ is stochastic reshuffling of $\left(E_{1}+E_{2},\dots,E_{n-1}\right)$. Let us consider the $n^{\text{th}}$ equation 
\begin{equation}
E_{n} = a_{n,\,1}\left(E_{1}+E_{2}\right)+\sum_{j=3}^{n} a_{n,\,j}E_{j}. \label{eq:En}
\end{equation}
If $a_{n,\,n}=1$ then since the operators are positive semi-definite, $a_{n,\,j}=a_{j,\,n}=0$ where $j < n$ for all $E_{j}\ne 0$. In the case $E_{j}=0$, we can replace the $j^{\text{th}}$ column of the matrix $a$ with any probability vector and \eqref{eq:mix} will still hold; hence we can assume a choice in which $a_{n,\,j}$ is zero. We therefore can write
\begin{equation}
E_{i} = a_{i,\,1}\left(E_{1}+E_{2} \right)+\sum_{j=3}^{n-1} a_{i,\,j}E_{j}, \label{eq:combine}
\end{equation}
for all $i\in\left\lbrace 1,\dots,n-1\right\rbrace$, where $a_{i,\,j}\geq 0$ and \mbox{$\sum_{i=1}^{n-1} a_{i,\,j}=1$} for all $j$. In other words $\left(E_{1},E_{2},\dots,E_{n-1}\right)$ is stochastic reshuffling of $\left(E_{1}+E_{2},\dots,E_{n-1}\right)$.

Alternatively, if $a_{n,\,n}<1$, we can rewrite \eqref{eq:En},
\begin{equation}
E_{n} = \frac{a_{n,\,1}}{1-a_{n,\,n}}\left(E_{1}+E_{2}\right)+\sum_{j=3}^{n-1} \frac{a_{n,\,j}}{1-a_{n,\,n}}E_{j}.
\end{equation}
Substituting this into \eqref{eq:mix}, we have
\begin{equation}
E_{i} = b_{i,\,1}\left(E_{1}+E_{2}\right)+\sum_{j=3}^{n-1} b_{i,\,j}E_{j},
\end{equation}
for all $i\in\left\lbrace 1,\dots,n-1\right\rbrace$, where 
\begin{equation}
b_{i,\,j} = a_{i,\,j} + a_{i,\,n}\frac{a_{n,\,j}}{1-a_{n,\,n}}.
\end{equation}
Clearly $b_{i,j}\geq 0$, and it is simple to show that \mbox{$\sum_{i=1}^{n-1}b_{i,\,j}=1$}. Thus \eqref{eq:mix} implies that $\left( E_{1},E_{2}, \dots, E_{n-1}\right)$ is a stochastic mixture of $\left( E_{1}+E_{2}, \dots, E_{n-1}\right)$. So we can apply this procedure $n-2$ times and conclude that $\left( E_{1},E_{2}\right)$ is a stochastic mix of $E_{1}+E_{2}$, which as we have already seen, implies $E_{1}\propto E_{2}$.

Any deterministic matrix can be considered a product of matrices whose action is to combine only two elements; hence if a stochastic operation combines a collection of POVM elements, any of which are not proportional to each other, then any subsequent stochastic operation cannot reverse this action.
\end{proof}

\section{Proof: The Buscemi measure of information gain is a resource monotones}\label{app:buscemi}

The Buscemi measure of information gain \cite{08buscemi210504} based upon the von Neumann entropy function, similar to the Maccone measure of information gain \cite{Maccone2007}. The proof that it is a resource monotone follows similarly to that of the Maccone measure, however we include it here for completeness.

The Buscemi measure is based upon an indirect measurement model consisting of four systems: the system $S$ being measured, a purification reference system $R$, a measurement apparatus $A$ and an environment $B$. The information gain is defined by considering the reduced density operator
\begin{equation}
\rho^{RA}=\sum_{k,l=1}^{d}\sum_{i=1}^{m} \sqrt{\lambda_{k}\lambda_{l}}\,\tr(F_{i}\ketbra{\psi_{k}^{S}}{\psi_{l}^{S}})\ketbra{\psi_{k}^{R}}{\psi_{l}^{R}}\otimes\ketbra{i^{A}}{i^{A}}.\label{eq:RA}
\end{equation}
where the system to be measured begins in the state \mbox{$\rho^{S}=\sum_{k=1}^{d} \lambda_{k} \ketbra{\psi_{k}^{S}}{\psi_{k}^{S}}$}.
The Buscemi measure of information gained is the quantum mutual information of this state
\begin{equation}
I_{\text{Bus}}\left(\rho,\vec{F}\right) = S\left(\rho^{R}\right) + S\left(\rho^{A}\right) - S\left(\rho^{RA}\right).
\end{equation}
Since $\rho^{RA}$ \eqref{eq:RA} is in a quantum-classical state, we can similarly rewrite the mutual information as
\begin{equation}
I_{\text{Bus}}\left(\rho,\vec{F}\right) = S\left(\rho^{R}\right) - \sum_{i=1}^{m} \Pr\left(F_{i}\right) S\left(\rho^{R(\vec{F})}_{i}\right),
\end{equation}
where $\Pr\left(F_{i}\right) = \tr(F_{i}\rho)$, and 
\begin{equation}
\rho^{R(\vec{F})}_{i} = \frac{1}{\Pr\left(F_{i}\right)}\sum_{k,\,l=1}^{d} \sqrt{\lambda_{k}\lambda_{l}}\,\tr(F_{i}\ketbra{\psi_{k}^{S}}{\psi_{l}^{S}})\ketbra{\psi_{k}^{R}}{\psi_{l}^{R}}.
\end{equation}
Since $\vec{F}$ is a stochastic reshuffling of $\vec{E}$, we can write $I_{\text{Bus}}$ using Bayes' theorem
\begin{align}
&\sum_{k,\,l=1}^{d}\sum_{j=1}^{n} \frac{\Pr\left(E_{j}|F_{i}\right)}{\Pr\left(E_{j}\right)}\sqrt{\lambda_{k}\lambda_{l}}\,\tr\left(E_{j}\ketbra{\psi_{k}^{S}}{\psi_{l}^{S}}\right)\ketbra{\psi_{k}^{R}}{\psi_{l}^{R}} \nonumber \\
&= \sum_{j=1}^{n} \Pr\left(E_{j}|F_{i}\right) \rho^{R(\vec{E})}_{j}.
\end{align}
The inequality results from the concavity of the von Neumann entropy,
\begin{align}
S\left(\rho^{R}\right) &- \sum_{i=1}^{m} \Pr\left(F_{i}\right) S\left(\rho^{R(\vec{F})}_{i}\right) \nonumber \\
&\leq S\left(\rho^{R}\right) - \sum_{i=1}^{n}\sum_{j=1}^{n} \Pr\left(F_{i}\right)\Pr\left(E_{j}|F_{i}\right) S\left(\rho^{R(\vec{E})}_{j}\right) \nonumber \\
&= S\left(\rho^{R}\right) - \sum_{j=1}^{n} \Pr\left(E_{j}\right)S\left(\rho^{R(\vec{E})}_{j}\right). \nonumber \\
&= I_{\text{Bus}}\left(\rho,\vec{E}\right).
\end{align}

\section{Proof: State discrimination games form a complete family of monotones}\label{app:games}

In this section we show that if \mbox{$\vec{E} = \left(E_{1},\dots,E_{n}\right)$} is more successful than \mbox{$\vec{F}=\left(F_{1},\dots,F_{m}\right)$} at discriminating between any ensemble of quantum states, then $\vec{E}$ can be transformed into $\vec{F}$ using free transformations.

A single quantum state discrimination game is defined with respect to a finite ensemble of quantum states,
\begin{equation}
\mathscr{E} = \left(q_{1}\rho_{1},\dots, q_{k}\rho_{k}\right),
\end{equation}
where
\begin{equation}
\tr{\left(\rho_{i}\right)} = 1 \text{ for all }i, \quad q_{j} \geq 0 \; \text{ for all }j, \quad \sum_{j=1}^{k} q_{j} = 1. \nonumber
\end{equation}
So the trace of $q_{i}\rho_{i}$ provides the probability $q_{i}$ that Alice will send Bob the state $\rho_{i}$.
The probability that Bob will successfully choose the correct state sent by Alice is
\begin{equation}
\expect{s_{\vec{E}}} = \sum_{i=1}^{n} \Pr\left(\rho_{i^{*}}^{\vec{E}}|E_{i}\right)\Pr\left(E_{i}\right), \label{eq:succ}
\end{equation}
where Bob's strategy is to choose the state from the ensemble with the highest posterior probability,
\begin{equation}
\rho_{i^{*}}^{\vec{E}} := \argmax_{\rho\in\set{\rho_{1},\dots,\rho_{k}}} \Pr\left(\rho|E_{i}\right).
\end{equation}
It was shown in section~\ref{sec:qsd} that the probability of success \eqref{eq:succ} for any particular ensemble is a resource monotone. However it was shown in \cite{Skrzypczyk2018} that the function \eqref{eq:succ} considered over the set of all possible ensembles forms a \emph{complete} family of monotones; we reproduce this proof in this appendix.

In fact, for this result we do not require that $\expect{s_{\vec{E}}}\geq \expect{s_{\vec{F}}}$ for all possible finite ensembles. We only only require that this holds for ensembles with the same number of elements as there are in the POVM $\vec{F}$; that is, $k=m$. This implicitly includes all ensembles with fewer elements, as these ensembles can contain up to $m-1$ terms with zero probability of being sent.

Let us here state von Neumann's minimax theorem \cite{V.Neumann1928}, which we intend to invoke.
\begin{theorem}[Minimax]\label{thm:minimax}
Let $X\subseteq \mathbb{R}^{\mu}$ and $Y\subseteq\mathbb{R}^{\nu}$ be compact convex sets. If $f:X\times Y \rightarrow \mathbb{R}$ is a continuous function with the following properties,
\begin{enumerate}
\item $f\left(\cdot,y\right):X\rightarrow \mathbb{R}$ is concave for fixed $y$,
\item $f\left(x,\cdot\right):Y\rightarrow \mathbb{R}$ is convex for fixed $x$,
\end{enumerate}
then
\begin{equation}
\min_{y\in Y} \max_{x\in X} f\left(x,y\right) = \max_{x\in X} \min_{y\in Y}  f\left(x,y\right).
\end{equation}
\qed
\end{theorem}

\begin{proposition}
Let \mbox{$\vec{E}=\left(E_{1}, \dots, E_{n}\right)$} and \mbox{$\vec{F}=\left(F_{1},\dots ,F_{m}\right)$} be POVMs. If $\expect{s_{\vec{E}}}\geq\expect{s_{\vec{F}}}$, where $\expect{s_{\vec{E}}}$ is defined as in \eqref{eq:succ},  for all ensembles of a $d$-dimensional quantum system containing $m$ elements, then $\vec{E}\geq \vec{F}$.
\end{proposition}
\begin{proof}
We begin by showing that given an ensemble \mbox{$\left(q_{1}\rho_{1},\dots, q_{m}\rho_{m}\right)$}, the probability of success $\expect{s_{\vec{E}}}$ \eqref{eq:succ} can be re-written
\begin{equation}
\expect{s_{\vec{E}}} = \max_{S} \sum_{i=1}^{m}\sum_{j=1}^{n} S_{i,j} \Pr\left(\rho_{i}|E_{j}\right)\Pr\left(E_{j}\right),
\end{equation}
where the maximisation is over all $m\times n$ stochastic matrices $S$. We can see that the maximand achieves the original definition \eqref{eq:succ} when $S_{i,j} = \delta_{i,j^{*}}$. To show that this is the maximal value, we have by definition \mbox{$\Pr\left(\rho_{i}|E_{j}\right)\leq \Pr\left(\rho_{j^{*}}^{\vec{E}}|E_{j}\right)$} for all states $\rho_{i}$ in the ensemble and all POVM elements $E_{j}$ in $\vec{E}$, so then for any stochastic matrix $S$,
\begin{align}
\sum_{i=1}^{m}\sum_{j=1}^{n} S_{i,j} &\Pr\left(\rho_{i}|E_{j}\right)\Pr\left(E_{j}\right) \nonumber \\ 
&\leq \sum_{i=1}^{m}\sum_{j=1}^{n} S_{i,j} \Pr\left(\rho_{j^{*}}^{\vec{E}}|E_{j}\right)\Pr\left(E_{j}\right)\nonumber \\
&= \sum_{j=1}^{n} \Pr\left(\rho_{j^{*}}^{\vec{E}}|E_{j}\right)\Pr\left(E_{j}\right).
\end{align}
The stochastic matrix $S$ can be thought to encode Bob's decision strategy. For example, if $m=n$, and we choose $S_{i,j}=\delta_{i,j}$, then we arrive at the probability of success whenever Bob guesses that Alice sent the state simply reported by his measurement device \eqref{eq:oldsa}.

Now suppose that $\expect{s_{\vec{E}}} \geq \expect{s_{\vec{F}}}$ for all ensembles containing $m$ elements. For any particular ensemble, this means
\begin{align}
\max_{S} &\sum_{i=1}^{m}\sum_{j=1}^{n} S_{i,j} \Pr\left(\rho_{i}|E_{j}\right)\Pr\left(E_{j}\right) \nonumber \\
&-\max_{S^{\prime}} \sum_{i=1}^{m}\sum_{j=1}^{m} S^{\prime}_{i,j} \Pr\left(\rho_{i}|F_{j}\right)\Pr\left(F_{j}\right)\geq 0.\label{eq:better}
\end{align}
This inequality  still holds if we choose the suboptimal choice of $S^{\prime}_{i,j} = \delta_{i,j}$. From quantum theory we have $\Pr\left(\rho_{i}|E_{j}\right)\Pr\left(E_{j}\right) = \tr{\left(q_{i}\rho_{i}E_{j}\right)}$, and \eqref{eq:better} becomes
\begin{align}
\max_{S}\sum_{i=1}^{m} q_{i} \tr{\left(\rho_{i}\left[\sum_{j=1}^{n} S_{i,j}E_{j} - F_{i}\right]\right)}\geq 0.\label{eq:cond}
\end{align}
Now let us consider the bracketed terms
\begin{equation}
\Delta_{i} = \sum_{j=1}^{n} S_{i,j}E_{j} - F_{i}.
\end{equation}
We note here that if there exists a stochastic matrix $S$ such that $\Delta_{i} = 0$ for all $i$, we have that $\vec{E} \geq \vec{F}$.  Let us now suppose that no stochastic matrix $S$ exists such that $\Delta_{i} = 0$ for all $i$, and show that this assumption leads to a contradiction. 
Now since \eqref{eq:cond} holds for all ensembles, it will hold under the minimisation over all ensembles
\begin{equation}
\min_{\mathscr{E}} \max_{S}\sum_{i=1}^{m} q_{i} \tr{\left(\rho_{i}\left[\sum_{j=1}^{n} S_{i,j}E_{j} - F_{i}\right]\right)}\geq 0. \label{eq:minmax}
\end{equation}

The set of $m\times n$ stochastic matrices is convex and is isomorphic to a subset of $\mathbb{R}^{mn}$. Likewise, the set of all $m$-element ensembles is convex, and isomorphic to a subset of $\mathbb{C}^{m d^{2}}\cong \mathbb{R}^{2md^{2}}$, where $d$ is the dimension of the quantum system being sent by Alice.
The maximand of \eqref{eq:cond} is concave over the set of ensembles, since for any $0\leq t\leq 1$,
\begin{align}
t \sum_{i=1}^{m} &\tr{\left(q_{i}\rho_{i}\Delta_{i}\right)}+(1-t)\sum_{i=1}^{m}\,\tr{\left(\tilde{q}_{i}\tilde{\rho_{i}}\Delta_{i}\right)} \nonumber\\
&=\sum_{i=1}^{m} \,\tr{\left(\left[t q_{i}\rho_{i}+(1-t)\tilde{q}_{i}\tilde{\rho_{i}}\right]\Delta_{i}\right)}.
\end{align}
It is also convex over the set of $m\times n$ stochastic matrices, due to the linearity of the trace function, and since
\begin{align}
t &\left(\sum_{j=1}^{n} S_{i,j}E_{j} - F_{i}\right) + (1-t)\left(\sum_{j=1}^{n} \tilde{S}_{i,j}E_{j} - F_{i}\right) \nonumber \\
&=\left(\sum_{j=1}^{n} \left(t S_{i,j} + (1-t) \tilde{S}_{i,j}\right) E_{j} - F_{i}\right).
\end{align}

Thus we can invoke the minimax theorem (theorem~\ref{thm:minimax})  to reorder the minimisation and maximisation in \eqref{eq:minmax}, after which we have
\begin{equation}
\max_{S}\min_{\mathscr{E}}\sum_{i=1}^{m} q_{i} \tr{\left(\rho_{i}\left[\sum_{j=1}^{n} S_{i,j}E_{j} - F_{i}\right]\right)}\geq 0.\label{eq:maxmin}
\end{equation}
Since $\vec{E}$ and $\vec{F}$ both sum to the identity, we have for any $S$,
\begin{equation}
\sum_{i=1}^{m} \Delta_{i} = \sum_{j=1}^{n} E_{j} - \sum_{i=1}^{m} F_{i} = 0.\label{eq:posSD}
\end{equation}
By assumption not all $\Delta_{i}$ are zero, and this with \eqref{eq:posSD} implies that all $\Delta_{i}$ cannot all be positive semi-definite: there must exist at least one term with a negative eigenvalue (call it $\Delta_{1}$),
\begin{equation}
\Delta_{1}\ket{\lambda} = -\lambda\ket{\lambda}.
\end{equation}
If we consider an ensemble whose only non-zero term is $\ketbra{\lambda}{\lambda}$, then
\begin{equation}
\sum_{i=1}^{m} q_{i}\tr{\left(\rho_{i} \Delta_{i} \right)} = -\lambda \braket{\lambda}{\lambda} < 0,
\end{equation}
which contradicts \eqref{eq:maxmin}. Hence it must be that $\Delta_{i}=0$ for all $i$, which implies that $\vec{E} \geq \vec{F}$.
\end{proof}

\end{document}